\documentclass[twocolumn,pra,aps,showpacs]{revtex4}
\usepackage{bm}
\usepackage{mathrsfs}
\usepackage{amsmath}
\usepackage{amssymb}
\usepackage{graphicx}
\usepackage{amsfonts}
\usepackage{amsthm}
\usepackage{color}
\usepackage{dcolumn}
\usepackage{txfonts}
\newtheorem{thm}{Theorem}
\newtheorem{defi}{Theorem}

\newtheorem{definition}[defi]{Definition}

\begin{document}

\title{Protection of quantum systems by nested dynamical decoupling}

\author{Zhen-Yu Wang}

\author{Ren-Bao Liu}
\email{rbliu@cuhk.edu.hk}

\affiliation{Department of Physics, The Chinese University of Hong
Hong, Shatin, N. T., Hong Kong, China}

\begin{abstract}
Based on a theorem we establish on dynamical decoupling of
time-dependent systems, we present a scheme of nested Uhrig
dynamical decoupling (NUDD) to protect multi-qubit systems in
generic quantum baths to arbitrary decoupling orders, using only
single-qubit operations. The number of control pulses in NUDD
increases polynomially with the decoupling order. For general
multi-level systems, this scheme can preserve a set of unitary
Hermitian system operators which mutually either commute or
anti-commute, and hence all operators in the Lie algebra generated
from this set of operators, generating an effective symmetry group
for the system up to a given order of precision. NUDD can be
implemented with pulses of finite amplitude, up to an error in the
second order of the pulse durations.
\end{abstract}


\maketitle

\section{Introduction}
Both in high-precision magnetic resonance spectroscopy and in
quantum computing, it is essential to suppress unwanted couplings
within a quantum system and between the system and its environment
(or bath). Such couplings result in population relaxation, phase
randomization (pure dephasing), and more generally, unwanted
evolution of certain system operators. Dynamical decoupling (DD) is
a so-called open-loop control scheme to average out the system-bath
interactions through stroboscopic operations of the system (without
direct control on the environment). DD originated from the Hahn
echo~\cite{Hahn1950Echo} and has evolved into variations with
complicated sequences~\cite{Carr1954CP} for
high-precision~\cite{Mehring_NMR,Schweiger} and
multi-dimensional~\cite{Ernst1987_Book} magnetic resonance
spectroscopy. When the field of quantum computing was opened up, DD
was introduced to protect qubit
coherence~\cite{Viola1998_PRA,Ban_DD}. By unitary symmetrization
procedure~\cite{Zanardi:99PLA,Viola1999_PRL}, DD cancels errors of
quantum evolutions up to the first order in the Magnus expansion,
and the corresponding cyclic scheme is referred to as periodic DD. A
geometric understanding of the symmetrization procedure was given in
Ref.~\cite{Byrd2002GeometricDD}. To eliminate errors to the second
order in the Magnus expansion, mirror-symmetric arrangement of two
DD sequences (SDD) can be used~\cite{Viola1999_PRL}. A particularly
interesting scheme is the concatenated DD
(CDD)~\cite{Khodjasteh2005_PRL,Khodjasteh:2007PRA,Yao2007_RestoreCoherence,Santos2008NJP},
which uses recursively constructed pulse sequences to eliminate
decoherence to an arbitrary decoupling order (defined as the power
of the total evolution time, which is assumed short). The
performance of CDD was experimentally demonstrated for spins in
solid-state environments~\cite{West_DDgate}. The number of pulses
used in CDD, however, increases exponentially with the decoupling
order. Since errors are inevitably introduced in each control pulse
in experiments, finding DD schemes with fewer control pulses is
desirable.

For suppressing pure dephasing of single qubits (two-level systems)
subjected to unidirectional noises, a remarkable advance is the
optimal DD discovered by Uhrig~\cite{Uhrig2007UDD} in a spin-boson
model. Uhrig DD (UDD) is optimal in the sense that the number of
control pulses is minimum for a given decoupling order. It was later
conjectured~\cite{Lee2008_PRL,UhrigNJP2008} and then rigorously
proven~\cite{Yang2008PRL} that UDD is model-independent for any
two-level systems coupled to a finite quantum bath. It was shown
that UDD also works for suppressing longitudinal
relaxation~\cite{Yang2008PRL}. The ideal $\delta$-pulses assumed in
DD can be generalized to include some components of finite
amplitude~\cite{Yang2008PRL}. Recently, a method to incorporate
shaped pulses of finite amplitude into UDD paves the way of
realistic experiments~\cite{Uhrig2010}. UDD was first verified in
experiments by microwave control of trapped ions in various
artificial classical
noises~\cite{Biercuk2009QuantumMemory,Biercuk2009a,Uys2009}, and
then UDD against realistic quantum noises was realized for radical
electron spins in irradiated malonic acid
crystals~\cite{Du2009Nature}.

For suppressing the general decoherence  of single qubits (including
both pure dephasing and longitudinal relaxation), the concatenation
of UDD sequences (CUDD) was proposed to reduce the number of control
pulses~\cite{uhrig:CUDD}. For suppressing the decoherence up to an
order $N$, the number of pulses required in CUDD is
$\sim(N+1)2^{N},$ which is considerably less than $\sim4^{N}$ as in
CDD. Recently, West \textit{et al} proposed a much more efficient
scheme, called quadratic DD (QDD), to combat general decoherence of
a qubit~\cite{West2010}. QDD is constructed by nesting two levels of
UDD sequences, using $(N+1)^2$ control intervals to achieve the
$N$th decoupling order. Numerical search indicates that QDD is
near-optimal as it differs from the optimal solutions by no more
than two pulses for a small decoupling order
($N\leq4$)~\cite{West2010}. The validity of UDD can be extended to
analytically time-dependent Hamiltonians~\cite{Pasini2010UDDTime}.
This extension seems to validate QDD since the UDD sequence on the
outer level can be viewed as acting on a time-dependent Hamiltonian
resulting from the UDD control on the inner level. However, as we
will show in Sec.~\ref{sec:Qubit}, the effective Hamiltonian
resulting from the inner UDD sequences is only piecewise analytic in
time, and actually there are counter examples in which UDD on the
outer level does not achieve the designed decoupling order if the
order of UDD on the inner level is {\em odd} and {\em lower} than
the order of the outer UDD. Thus it remains an open question why QDD
works. In the attempt to prove the validity of QDD, we establish a
theorem: UDD applies to time-dependent Hamiltonians, regardless of
their analytic properties, as long as the Hamiltonians between two
adjacent pulses of the UDD sequence are symmetric and have the same
function form of relative time between the adjacent pulses (see
Sec.~\ref{sub:NestedDD}). Therefore we give a proof of the validity
of QDD with even order UDD on the inner level. We note that the
validation of QDD is still {incomplete} since the theorem mentioned
above does not apply to QDD with odd order UDD on the inner level.

So far, the research of optimal DD focuses on the single-qubit
decoherence problem, with some attempts on optimal DD to multi-level
systems with prior knowledge of the initial
states~\cite{Dhar2006_PRL,MultilevelUDD:2010}. For practical
large-scale quantum computing, the issue is the decoherence, or more
generally, the decay of quantum correlations (such as entanglement)
of coupled multi-qubit systems. Therefore it is highly desirable to
have a general arbitrary-order DD scheme for multi-qubit systems,
with the number of pulses as small as possible. In addition, the DD
scheme should preferably involve relatively easy implementations,
for example, single-qubit operations. Realizations of the
symmetrization procedure~\cite{Zanardi:99PLA,Viola1999_PRL} by
averaging over so-called {\em nice} error
bases~\cite{Knill1996ErrorBases} are explicitly given in
Ref.~\cite{Wocjan02QIC}. However, each control operation in general
is complicated and may involve manipulation on all qubits for
multi-qubit systems. It would be of practical interest if there is
an explicit, systematic, and efficient way to protect a particular
set of operators of a quantum system to an arbitrary order.

In this paper, we give systematic and explicit DD schemes to protect
multi-qubit systems arbitrarily coupled to quantum baths. The
schemes are realized by protecting a set of mutually commuting or
anti-commuting unitary Hermitian system operators on different
levels. We call it a mutually orthogonal operation set (MOOS) in
this paper. For example, the Pauli operators of qubits form an MOOS.
The inner levels of DD control of operators in an MOOS are not
affected by the outer levels of DD control. Furthermore, based on
the above-mentioned theorem on DD of time-dependent systems, higher
order protection of an MOOS can be achieved by nesting even-order
UDD sequences on different levels. If a set of system operators is
protected by such nested UDD (NUDD), then all system operators in
the Lie algebra generated from this set of operators are protected
to the same decoupling order, which indeed generates an effective
symmetry group~\cite{Zanardi:99PLA} of the system up to an error of
the decoupling order. For multi-qubit systems, each control
operation in DD only involves single-qubit manipulation. In
addition, we will show that NUDD can be implemented with pulses of
finite amplitude, which approximate ideal $\delta$-pulses up to an
error in the second order of the pulse durations, with the same
pulse shaping as in Ref.~\cite{PasiniPRA2008}. For a general
multi-level quantum system, we can also construct an MOOS and use
NUDD to generate an effective symmetry group to a given decoupling
order. It can be shown, however, that for a general $M$-level
system, there may exist no MOOS to generate the whole basis of
su($M$) algebra and hence the whole SU($M$) symmetry group. Further
research is still needed to design efficient DD schemes (as compared
with CDD) to protect general multi-level systems to higher orders.

This paper is organized as follows. In
Sec.~\ref{sec:Preserving-arbitrary-quantum}, we present a general
theory on protection of an MOOS by DD; NUDD is given based on a
theorem established for UDD on time-dependent systems. In this
Section, the pulses in DD are assumed instantaneous. In
Sec.~\ref{sec:Qubit}, we discuss NUDD on multi-qubit systems. In
Sec.~\ref{sec:SoftPulse}, we discuss DD with finite-amplitude
pulses. Finally, we draw the conclusions in
Sec.~\ref{sec:Conclusions}.

\section{Protection of system operators\label{sec:Preserving-arbitrary-quantum}}

\subsection{General formalism: MOOS\label{sub:MOOS}}
We consider a quantum system coupled to a general finite quantum
bath, with a time-independent Hamiltonian
\begin{equation}
H=H_S+H_B+H_{SB},
\end{equation}
where $H_S$ is the system Hamiltonian, $H_B$ the bath one, and
$H_{SB}$ the system-bath interaction. We aim to
find a sequence of stroboscopic operations $O_1$, $O_2$, $\ldots$,
$O_q$ at times $T_1$, $T_2$, $\ldots$, $T_q$ in increasing order
so that after the controlled evolution from $t=0$ to $t=T$, a set of system operators
$$\{Q_j\}\equiv\{Q_1,Q_2,\ldots\},\nonumber$$ are conserved up to an
error of $O\left(T^{N+1}\right)$, i.e., for {($\hbar=1$)}
\begin{equation}
U\equiv e^{-iH(T-T_q)}O_q\cdots O_2e^{-iH(T_2-T_1)}O_1e^{-iHT_1},
\end{equation}
we have
\begin{equation}
U^{\dag}Q_jU=Q_j+O\left(T^{N+1}\right).
\end{equation}
In this paper we assume that the bath is bounded in spectrum so that
the perturbation expansion of $U$ in a short time $T$ is possible.
To make DD efficient, it is required that the evolution of the
system induced by the decoupling field be faster than the unwanted
dynamics~\cite{Viola1998_PRA}. Here we assume that the control
pulses are instantaneous and arbitrarily strong. Using pulses of
finite duration and finite amplitude will be discussed in
Sec.~\ref{sec:SoftPulse}. Preferably, for a multi-qubit system, the
operations $O_j$ should contain only single-qubit operations. We
also wish to use as few as possible operations to achieve a given
order ($N$) of decoupling precision.

We note that when a set of operators $\{Q_j\}$ is preserved to a
certain decoupling order, then all operators obtained by commutation
$i[Q_j,Q_k]$, anti-commutation $[Q_j,Q_k]_+\equiv Q_jQ_k+Q_kQ_j$,
linear combinations, and their repetitions are also protected to the
same decoupling order. All these protected operators form a Lie
algebra, since they form a linear vector space and are close under
the commutation operation. The Lie algebra defines a dynamically
generated effective symmetry group ${\mathcal S}_{Q}$ of the system
up to the decoupling order.

In particular, let us consider the protection of a set of operators
$\{\Omega_j\}$ in which each pair of elements either commutes or
anti-commutes. In our schemes, it is required that the operators
$\{\Omega_j\}$ be unitary and satisfy $\Omega_j^2=\pm1$. Thus we
choose $\Omega_j$ to be unitary and Hermitian, i.e.,
\begin{equation} \Omega_j^2=\Omega^{\dag}_j\Omega_j=1. \label{UH}
\end{equation}
Note that $\Omega_j$ is a parity kick operator described in
Ref.~\cite{Vitali1999PRA}. We will use the unitary Hermitian
property of $\Omega_j$ to construct DD for protection of $\Omega_j$.
\begin{definition}{An MOOS is defined as a set of operators which are unitary and
Hermitian and have the property that each pair of elements either
commutes or anti-commutes.}
\end{definition}

The commutation property of operators in an MOOS is important for
constructing higher-order DD schemes, via the following theorem.
\begin{thm}
\label{thm:MOOS} For two operators $\Omega_1$ and $\Omega_2$ either
commuting or anti-commuting with each other, if a unitary evolution
$U(T)$ during a short time $T$ commutes with $\Omega_1$ up to an
error of $O\left(T^{N+1}\right)$, then $\Omega_2U(T)\Omega_2$ also
commutes with $\Omega_1$ up to an error of $O\left(T^{N+1}\right)$.
\end{thm}
\begin{proof}
We directly calculate the commutator
$$\left[\Omega_1,\Omega_2U(T)\Omega_2\right]=(-1)^\eta\Omega_2\left[\Omega_1,U(T)\right]\Omega_2
=O\left(T^{N+1}\right),$$ where $\eta=+1$ ($-1$) for $\Omega_1$
commuting (anti-commuting) with $\Omega_2$.
\end{proof}
With this theorem, certain DD sequences protecting a given operator
in an MOOS can be used as units to construct an outer level of DD
protection of another operator in the MOOS, without affecting the DD
effect of the inner level control. For operators which do not form
an MOOS, the outer level control in general may interfere with the
inner level control. For example, let us consider
$\Omega_1=\sigma_x$ and
$\Omega_2=\left(\sigma_x+\sigma_y\right)/\sqrt{2}$ (the Pauli matrix
along the direction equally dividing the angle between the $x$- and
$y$- axes). Suppose a DD sequence protects $\sigma_x$ as
$U(T)=e^{-i\sigma_x T+ O(T^{N+1})}$. Then we apply a Hahn echo to
protect $\Omega_2$, by the following evolution
$$
U_2(2T)=\Omega_2U(T)\Omega_2U(T)=e^{-i\sigma_y T+
O(T^{N+1})}e^{-i\sigma_x T+ O(T^{N+1})}.
$$
After the control on the outer level, the evolution actually does
not commute with $\sigma_x$ even in the leading order of $T$.

As we will show in Sec.~\ref{sub:FirstOrder}, applying an element of
an MOOS corresponds to a symmetrization
procedure~\cite{Zanardi:99PLA,Viola1999_PRL} on one level, and
protecting an MOOS will iteratively symmetrize the system evolution.
An MOOS itself does not form a group or an algebra. However, it can
generate a certain Lie algebra by commutation, anti-commutation,
linear combinations, and repetitions. Therefore the protection of
all operators in this Lie algebra is realized by the protection of
this MOOS. For example, in a single-qubit system, $N$ pulses of
$\pi$-rotation control of the qubit about the $z$-axis, when
arranged according to the UDD timing, protect the Pauli matrix
$\sigma_z$ and dynamically generate the symmetry group
$\text{SO}(2)$ or $\text{U}(1)$, up to an error of
$O(T^{N+1})$~\cite{Yang2008PRL}. For a pure dephasing Hamiltonian,
which has $\text{U}(1)$ as its intrinsic symmetry group, the system
under the UDD protection has the full $\text{SU}(2)$ symmetry.
Another example is the $N$th order QDD sequence~\cite{West2010}
consisting of pulses of $\sigma_x$ and $\sigma_y$, which protects
$\sigma_x$ and $\sigma_y$ and hence all operators in the
$\text{su}(2)$ algebra (including $\sigma_z$), and therefore the
qubit has a dynamically generated symmetry group $\text{SU}(2)$, up
to an error of $O\left(T^{N+1}\right)$.

Below we explicitly construct the MOOS for some particular Lie
algebra to be protected:
\begin{enumerate}
\item
For suppressing longitudinal relaxation of a single-qubit system
along the $z$-direction, the MOOS contains only one operator
$\sigma_z$.
\item
For a general single-qubit system, two anti-commuting Pauli
matrices, e.g., $\{\sigma_x,\sigma_y\}$, form an MOOS to protect all
system operators~\cite{Duan:1998km}.
\item
For an $L$-qubit pure dephasing model in which $H$ contains only the
Pauli matrices along the $z$-direction, a choice of the MOOS is
$$\{\sigma_{x}^{(l)}\}_{l=1}^{L} \equiv \{\sigma_{x}^{(1)}, \sigma_{x}^{(2)}, \ldots, \sigma_{x}^{(L)}\}.$$
Note that only single-qubit operators are used. It is obvious that
all the Pauli matrices in the set commute with each other. There are
totally $L$ operators in the MOOS, which can be shown to be the
minimum possible number.
\item
For a general $L$-qubit system, to protect all operators of the
system, i.e., the algebra $\text{su}(2^L)$, a choice of the MOOS is
$\{\sigma^{(l)}_z, \sigma^{(l)}_x\}_{l=1}^{L}$. It is obvious that
all the Pauli matrices either commute or anti-commute with each
other. Note that only single-qubit operators are used. There are
totally $2L$ operators in the MOOS, which can be shown to be the
minimum possible number.
\item
For a general $M$-level system, let us consider protection of all
system operators which are diagonal in a certain orthonormal basis
$\{|m\rangle\}_{m=0}^{M-1}$. Suppose $2^{L-1}<M\le 2^L$. We can
denote an integer number $0\le m <M$ using a binary code as
$m=(m_{L}\cdots m_{2}m_{1})$ with $m_l=0$ or $1$. We define a
unitary Hermitian operator
\begin{align}
\Sigma_z^{(l)}=I-2\sum_{m_l=1}|m\rangle\langle m|.
\label{Sigmaz}
\end{align}
    A diagonal operator of the form $|m\rangle\langle m|$ can be written in terms of $\{\Sigma^{(l)}_z\}$ as
    $$|m\rangle\langle m|=\prod_{l=1}^L\frac{I+(-1)^{m_l}\Sigma^{(l)}_z}{2}.$$
    Thus any diagonal operators can be constructed using $\{\Sigma^{(l)}_z\}$.
    Actually, by adding $(2^L-M)$ ancillary basis states $\{|m\rangle\}_{m=M}^{2^L-1}$,
    the above-defined operators can be viewed as the single-qubit Pauli matrices along the $z$-direction
    ($\{\sigma^{(l)}_z\}$) of an $L$-qubit system projected to the $M$-level subsystem. It is obvious that all operators in $\{\Sigma^{(l)}_z\}$ commute with each other.
    Thus an MOOS is constructed. There are totally $L$ such operators in the MOOS.
    Since an $(L-1)$-qubit system has at least $(L-1)$ operators in an MOOS and $2^{L-1}<M \leq 2^L$, the MOOS
    $\{\Sigma^{(l)}_z\}_{l=1}^{L}$ contains the minimum number of operators (otherwise,
    we can construct a DD with less than $M$ intervals to protect the system to the first decoupling order, which is
    impossible according to Ref.~\cite{Wocjan02QIC}).
\item
For a given $l$, if $M/2^{l}$ is an integer, we can define unitary
Hermitian operators anti-commuting with $\Sigma_z^{(l)}$ as
\begin{align}
 \Sigma_x^{(l)} =& \sum_{m_l=0}\left(|m+2^{l-1}\rangle\langle m|+ \text{h.c.} \right),
\label{Sigmax}
\end{align}
which exchanges two basis states $|m\rangle$ and $|m'\rangle$ if $m$
and $m'$ differ at and only at the $l$th bit. Actually, the operator
$\Sigma^{(l)}_x$ can be viewed as the Pauli matrix $\sigma^{(l)}_x$
of the $l$th qubit for a multi-qubit system.
$\{\Sigma_x^{(l)}|M\text{~mod~}2^l=0 \} \cup \{\Sigma_z^{l}|2^l\le
M\}$ forms an MOOS of the $M$-level system.
\end{enumerate}

\noindent Note 1: The choice of an MOOS for a certain system is not
unique. For example, in the MOOS for generating $\text{su}(2^L)$ of
a general $L$-qubit system, the two Pauli matrices $\sigma^{(l)}_z$
and $\sigma^{(l)}_x$ can be replaced with any two anti-commuting
Pauli matrices of the $l$th qubit.

\noindent Note 2: For a general $M$-level system, there may exist no
MOOS to generate the whole su($M$) algebra. For example, for an
$M$-level system with $M$ being an odd number, all operators in an
MOOS must mutually commute and therefore are all diagonal in a
common basis. This MOOS of course cannot generate the whole su($M$).
Actually, if two operators $\Omega$ and $\Omega'$ in an MOOS
anti-commute, we have
${\text{Tr}}\left(\Omega+\Omega'\Omega\Omega'\right)=2\text{Tr}(\Omega)=0$,
which is impossible for odd $M$.

For a general $M$-level system, the explicit operations given in
Eqs.~(\ref{Sigmaz}) and (\ref{Sigmax}) are in general difficult to
implement in experiments. It is important to find a suitable set of
operations for a given multi-level system. It should also be noted
that above we have assumed that all the operators in an MOOS are
protected to the same decoupling order. In the CDD and NUDD schemes
we will discuss later, different operators actually can be protected
to different orders. In those cases, the precision of protection of
the whole Lie algebra is determined by the lowest decoupling order.

\subsection{Lowest order protection of system operators\label{sub:FirstOrder}}

A general DD framework to protect a set of operators to the lowest
order is the symmetrization procedure over an appropriate DD
group~\cite{Zanardi:99PLA,Viola1999_PRL}. Here we systematically
give an explicit scheme to protect system operators forming an MOOS,
which facilitates the construction of higher order DD in the later
part of this paper.

Let us consider first protection of a single unitary Hermitian
operator $\Omega$. The Hamiltonian can be separated into two parts,
\begin{equation} H=C_{\Omega}+A_{\Omega},
\end{equation}
with
\begin{subequations}
\begin{align}
C_{\Omega}\equiv(H+\Omega H\Omega)/2,\\
A_{\Omega}\equiv(H-\Omega H\Omega)/2.
\end{align}
\end{subequations}
$\Omega$ commutes with $C_{\Omega}$ and anti-commutes with
$A_{\Omega}$, i.e.,
\begin{subequations}
\begin{align}
\Omega C_{\Omega}\Omega=C_{\Omega},\\
\Omega A_{\Omega}\Omega=-A_{\Omega}.
\end{align}
\end{subequations}

With an instantaneous control pulse $\Omega$ applied at the middle of
the evolution time, the evolution operator becomes
\begin{align}
U_{\Omega}(T) & =e^{-iHT/2}\Omega e^{-iHT/2}=\Omega e^{-i\Omega H\Omega T/2}e^{-iHT/2}
\nonumber \\
& =\Omega e^{-iC_{\Omega}T}+O(T^{2}),
\end{align}
which commutes with $\Omega$ up to an error of $O(T^{2})$. In
general, one may apply an additional pulse at the end of evolution
so that
\begin{equation}
[\Omega] U_{0}(T/2)\Omega U_{0}(T/2)\equiv e^{-iH_{\Omega}T},
\label{eq:Uomega1st}
\end{equation}
where $U_{0}(\tau)\equiv e^{-iH\tau}$ is the free evolution operator
over time $\tau$, and the brackets around the operation at the end
of the sequence ($[\Omega]$) mean that the operation is optional.
The effective Hamiltonian $H_{\Omega}$ commutes with $\Omega$ up to
$O(T)$. Therefore $\Omega$ is protected to the first order.

Following the method given above, we can preserve more operators
$\{\Omega_{k}\}$ in an MOOS by concatenation. The first level of
control is \begin{equation}
U_{1}(T)=[\Omega_{1}]U_{0}(T/2)\Omega_{1}U_{0}(T/2)=e^{-iH_{1}T},\label{eq:U1T}\end{equation}
 where the effective Hamiltonian $H_{1}$ commutes with $\Omega_{1}$
up to $O(T)$. By treating the effective Hamiltonian as a new
Hamiltonian on the second level, the propagator reads
\begin{equation}
U_{2}(T)=[\Omega_{2}]U_{1}(T/2)\Omega_{2}U_{1}(T/2)=e^{-iH_{2}T},\label{eq:U2T}
\end{equation}
where the effective Hamiltonian $H_{2}$ commutes with $\Omega_{2}$
up to an error of $O(T)$. And according to Theorem~\ref{thm:MOOS},
$H_{2}$ also commutes with $\Omega_1$ up to $O(T)$.  A general first
order scheme is achieved by using Eq.~(\ref{eq:Uomega1st})
iteratively,
\begin{equation}
U_{L}(T)=[\Omega_{L}]U_{L-1}(T/2)\Omega_{L}U_{L-1}(T/2)=e^{-iH_{L}T},
\label{eq:UkT1}
\end{equation}
where the effective Hamiltonian $H_{L}$ commutes with all operators
in the MOOS $\{\Omega_{l}\}_{l=1}^{L}$ up to $O(T)$. Note that when
the optional pulses are not used, only one pulse is applied at each
time of operation.

It should be pointed out that the current DD scheme constructed by
iteration is not the same as the symmetrization procedure described
in Refs.~\cite{Zanardi:99PLA} and \cite{Viola1999_PRL}. For
$M$-level systems, when $\log_{2}M$ is not an integer, our scheme in
general cannot achieve protection of all system operators. For
example, when $M$ is odd, our scheme can only protect operators
diagonal in a certain basis. When $M=2^L$, all system operators can
be protected by protecting the MOOS of size $2L$. Several advantages
of the MOOS-based DD, however, are worth mentioning. First, there
are systematic ways to construct higher order DD for protecting all
operators in an MOOS and hence the Lie algebra generated from the
MOOS (Secs.~\ref{sub:GeneralCDD} and \ref{sub:NestedDD}). Second,
for multi-qubit systems, our scheme automatically gives DD sequences
involving only single-qubit operations (Sec.~\ref{sec:Qubit}).
Third, there is an explicit scheme to incorporate pulses of finite
amplitude to DD (Sec.~\ref{sec:SoftPulse}). Fourth, since the DD
sequences on the inner levels use more control pulses than the
sequences on the outer levels, qubits subjected to faster error
sources can be protected on inner levels for economic use of control
resources. For example, in a coupled electron-nuclear spin system,
the electron spin, which has much faster decoherence than a nuclear
spin does, should be controlled on the inner level.

\subsection{Higher order protection by CDD\label{sub:GeneralCDD}}

To go beyond the lowest order protection, let us consider a general
sequence of unitary operations $\{\sigma_{k}\}$ on the system, where
$\sigma_{k}$ can be the identity operator. The evolution from $t=0$
to $T$ reads
\begin{eqnarray}
U_{C}(T) & = & \sigma_{k_{n}}^{\dagger}e^{-iH(T-t_{n})}\sigma_{k_{n}}\cdots\sigma_{k_{1}}^{\dagger}e^{-iH(t_{2}-t_{1})}\sigma_{k_{1}}\sigma_{k_{0}}^{\dagger}e^{-iHt_{1}}\sigma_{k_{0}}\nonumber \\
 & \equiv & \mathcal{T}\exp\left[-i\int_{0}^{T}H_{\sigma}(t)dt\right],
\end{eqnarray}
where $\mathcal{T}$ is the time ordering operator, and
$H_{\sigma}(t)\equiv\sigma_{k_{j}}^{\dagger}H\sigma_{k_{j}}$ for
$t\in(t_{j},t_{j+1}]$. In the standard time-dependent perturbation
theory formalism, the propagator is expanded up to the second order
as
\begin{equation}
U_{C}(T)=1+\sum_{i}h_{i}+\sum_{i>j}h_{i}h_{j}+\sum_{i}\frac{1}{2}h_{i}^{2}+O(T^{3}),
\label{eq:UsigmaExpansion}
\end{equation}
 where $h_{j}\equiv-i(t_{j+1}-t_{j})H_{\sigma}(t)$ (with $t_0\equiv 0$ and $t_{n+1}\equiv T$).
In this paper we assume that this expansion converges, which means
that the bath is bounded in spectrum.

It was shown that if $U_{C}(T)$ realizes the first order DD, i.e.,
$[\sum_{i=0}^{n} h_i , Q_j]=0$ for a set of operators $\{Q_j\}$,
then the second order DD can be realized by the symmetrized
evolution~\cite{Ernst1987_Book,Viola1999_PRL} \begin{align}
U_{\text{SDD}}(2T)&\equiv \bar{U}_{C}U_{C} \nonumber\\
&=1+2\sum_{i=0}^{n}h_{i}+\frac{1}{2!}\left(2\sum_{i=0}^{n}h_{i}\right)^{2}+O(T^{3}),
\end{align}
where $\bar{U}_{C}\equiv e^{h_{0}}e^{h_{1}}\cdots
e^{h_{n-1}}e^{h_{n}}$ is mirror-symmetric with $U_{C}$.

For DD to even higher orders, in principle we can obtain the optimal
sequences by solving Eq.~(\ref{eq:UsigmaExpansion}) so that up to
some order in the expansion, $U_{C}(T)$ commutes with a given set of
operators. The third order terms read
$$ \sum_{i>j>k}h_{i}h_{j}h_{k}+\sum_{i>j}\frac{1}{2!}h_{i}^{2}h_{j}+\sum_{i>j}\frac{1}{2!}h_{i}h_{j}^{2}+\sum_{i}\frac{1}{3!}h_{i}^{3},
$$
and in general
$\sum_{j_{1}>j_{2}>\cdots>j_{n}}h_{j_{1}}^{p_{1}}h_{j_{2}}^{p_{2}}\cdots
h_{j_{n}}^{p_{n}}\prod_{r=1}^{n}\frac{1}{p_{r}!}$ contains terms of
the order $\sum_{r=1}^{n}p_{r}$. Finding solutions becomes
formidable when the DD order is high.

If we are not concerned with the exponentially increasing number of
control pulses, we can follow the idea of
CDD~\cite{Khodjasteh:2007PRA,Santos2008NJP} to construct DD
sequences in a systematic way to protect the operators
$\{\Omega_{l}\}_{l=1}^{L}$ to an arbitrary DD order. The first order
DD given by Eq.~(\ref{eq:UkT1}) is $2^{L}$ evolution operators
$U(T/2^{L})$ embedded in a sequence of control pulses
$\{\Omega_{l}\}$. We denote this structure as
\begin{align}
U_{L}^{[1]}(T)\equiv&\Omega_{L}U_{L-1}(T/2)\Omega_{L}U_{L-1}(T/2)\nonumber\\
\equiv&\mathcal{C}_{\Omega}\left\{U(T/2^{L})\right\}=e^{-iH_{L}^{[1]}T}\label{eq:CDDway1UL1},
\end{align}
which is the first order CDD as defined in Eq.~(\ref{eq:UkT1}). The
resultant first-order effective Hamiltonian $H_L^{[1]}$ commutes
with the operators $\{\Omega_{l}\}_{l=1}^{L}$ up to an error of
$O(T)$. Here the sequence
$\mathcal{C}_{\Omega}\left\{\cdots\right\}$ makes the effective
Hamiltonian commute with $\{\Omega_{l}\}_{l=1}^{L}$ to a higher
order. The evolution under the second order CDD
\begin{equation}
U_{L}^{[2]}(T)=\mathcal{C}_{\Omega}\left\{U_{L}^{[1]}(T/2^{L})\right\}=e^{-iH_{L}^{[2]}T},
\end{equation}
is obtained by replacing the free evolution $U(T/2^L)$ in
Eq.~(\ref{eq:CDDway1UL1}) with $U_{L}^{[1]}(T/2^{L})$, an evolution
operator under the first order CDD control. The resultant effective
Hamiltonian $H_{L}^{[2]}$ commutes with $\{\Omega_{l}\}_{l=1}^{L}$
up to an error of $O(T^2)$. Iteratively, the $N$th order CDD
reads\begin{equation}
U_{L}^{[N]}(T)=\mathcal{C}_{\Omega}\left\{U_{L}^{[N-1]}(T/2^{L})\right\},\label{eq:CDDway1}
\end{equation}
which preserves any operators in $\{\Omega_{l}\}_{l=1}^{L}$ up to
$O(T^{N+1})$. Note that in the above construction of CDD,
Theorem~\ref{thm:MOOS} has not been invoked, except in construction
of the innermost level as in Eq.~(\ref{eq:UkT1}). The only
requirement is that all the intervals of the outer level are equal
so that the resultant effective Hamiltonian from the inner level of
control is time-independent.

An alternative construction of CDD is given as follows. We first
construct a CDD sequence to protect $\Omega_{1}$ up to an error of
$O\left(T^{N_1+1}\right)$ by recursion
\begin{equation}
U_{\Omega_1}^{\text{C}[N_{1}]}(T)=\Omega_{1}U_{\Omega_1}^{\text{C}[N_{1}-1]}(T/2)\Omega_{1}U_{\Omega_1}^{\text{C}[N_{1}-1]}(T/2),
\end{equation}
with $U_{\Omega_1}^{\text{C}[0]}(T)\equiv U(T)$ and the superscript
$\text{C}$ denoting the nesting scheme of CDD. By defining
$U_{\Omega_1,\Omega_2}^{\text{C}[N_{1},0]}(T)\equiv
U_{\Omega_1}^{\text{C}[N_{1}]}(T)$, we can construct a further level
of CDD to protect $\Omega_{2}$ up to $O\left(T^{N_2+1}\right)$, by
the recursion
\begin{equation}
U_{\Omega_1,\Omega_2}^{\text{C}[N_{1},N_{2}]}(T)=\Omega_{2}U_{\Omega_1,\Omega_2}^{\text{C}[N_{1},N_{2}-1]}(T/2)\Omega_{2}
U_{\Omega_1,\Omega_2}^{\text{C}[N_{1,}N_{2}-1]}(T/2).
\end{equation}
Similarly, we have the propagator by recursion
\begin{subequations}
\begin{align}
U_{\Omega_1,\Omega_2,\ldots,\Omega_l,\Omega_l+1}^{\text{C}[N_{1},N_2,\ldots,N_{l},0]}(T) &\equiv  U_{\Omega_1,\Omega_2,\ldots,\Omega_l}^{\text{C}[N_{1},N_2,\ldots,N_{l}]}(T), \\
U_{\Omega_1,\Omega_2,\ldots,\Omega_L}^{\text{C}[N_{1},\ldots,N_{L}]}(T)
=&
\Omega_{L}U_{\Omega_1,\Omega_2,\ldots,\Omega_L}^{\text{C}[N_{1},\ldots,N_{L-1},N_{L}-1]}(T/2)\Omega_{L}
\nonumber \\ &\times
U_{\Omega_1,\Omega_2,\ldots,\Omega_L}^{\text{C}[N_{1},\ldots,N_{L-1},N_{L}-1]}(T/2).
\end{align}
\label{eq:CDDway2}
\end{subequations}
According to Theorem~\ref{thm:MOOS}, the inner levels of DD are
unaffected by the outer levels of control. Thus the evolution
$U_{\Omega_1,\Omega_2,\ldots,\Omega_L}^{\text{C}[N_{1},N_{2},\ldots,N_{L}]}(T)$
commutes with $\Omega_{1}, \Omega_{2}, \ldots, \Omega_{L}$ to the
orders $N_{1}, N_{2}, \ldots, N_{L}$, in turn. An advantage of this
construction is that the errors induced by $\Omega_{l}$ are
eliminated independently and to different orders $\{N_{l}\}$, which
allows protecting operators with stronger error sources to higher
orders. For example, usually for spin qubits under strong external
magnetic field, the pure dephasing is much faster than the
population relaxation, so it is favorable to protect the phase
correlation on the inner level and to a higher CDD order.

The number of operations required in preserving the operators is
$\sim 2^{NL}$ for the CDD scheme in Eq.~(\ref{eq:CDDway1}) or $\sim
2^{\sum_{l=1}^{L}N_{l}}$ for that in Eq.~(\ref{eq:CDDway2}). They
increase exponentially with the DD order. Even though the
exponentially increasing number of control pulses does yield
significant improvement of precision (through reduction of the
coefficient in front of the power of time
$T^{N+1}$)~\cite{Khodjasteh2005_PRL,Liu2007_NJP,Lidar_DDgate},
implementation of CDD to high orders is challenging in experiments
since errors are inevitably introduced in each control pulse.

\subsection{Higher order protection by NUDD\label{sub:NestedDD}}

If there is only one unitary Hermitian operator $\Omega$ to be
preserved, $N$ operations of $\Omega$ applied at the UDD timing~\cite{Uhrig2007UDD}
\begin{equation}
T_{n}=T \sin^{2}\frac{n\pi}{2N+2},\text{~for~}n=1,\ldots,
N,\end{equation} during the evolution from $T_{0}\equiv 0$ to
$T_{N+1}\equiv T$ protect the physical quantity $\Omega$ to the
$N$th order~\cite{Yang2010ReviewDD}. Explicitly, the propagator
under control
\begin{equation}
U_{\Omega}^{\text{U}[N]}(T)\equiv\Omega^{N}U_{0}(\tau_{N})\cdots\Omega
U_{0}(\tau_{1})\Omega U_{0}(\tau_{0}),
\end{equation}
commutes with $\Omega$ up to an error of $O(T^{N+1})$. Here the
evolution intervals are
\begin{equation}
\tau_{n}\equiv
T_{n+1}-T_{n}=\frac{T}{2}\left[\cos\frac{n\pi}{N+1}-\cos\frac{(n+1)\pi}{N+1}\right].\label{eq:UDDtiming}\end{equation}
When there are more than one unitary Hermitian operators
$\{\Omega_{l}\}_{l=1}^L$ to be protected, the question is whether we
can construct NUDD so that the number of control pulses scales
polynomially with the protection order. A known example is QDD, in
which UDD of, e.g., $\sigma_{z}$ and $\sigma_{x}$ are nested. For
general cases, we establish the following theorem as the basis of
NUDD.

\subsubsection{A theorem on UDD control of time-dependent {systems}}

\begin{thm}
\label{thm:Ht}
For a finite-norm time-dependent Hamiltonian $H(t)$
defined in $[0,T]$, an $N$th order UDD control with $N$ operations
of unitary Hermitian operator $\Omega$ applied at $T_{1}$, $T_{2}$,
$\ldots$, $T_{N}$ preserves $\Omega$ up to an error of
$O\left(T^{N+1}\right)$, if\begin{align}
H\left(T_{n}+s\tau_{n}\right)=H\left(T_{n+1}-s\tau_{n}\right)=H(sT_{1}),
\label{eq:HtSCondit}
\end{align}
for $s\in[0,1]$ and $\tau_{n}=T_{n+1}-T_{n}$, i.e., the Hamiltonian
has the same form as a function of the relative time between
adjacent operations and is symmetric within each interval.
\end{thm}

\noindent Note: The previous extension of UDD to time-dependent
systems requires that the Hamiltonians be analytic (having smooth
time-dependence)~\cite{Pasini2010UDDTime,MultilevelUDD:2010,Yang2010ReviewDD}.
In Theorem~\ref{thm:Ht}, the Hamiltonians are not required to be
analytic but with certain symmetries. The symmetry requirements on
the time-dependence of the Hamiltonians can actually be fulfilled by
designing the timing of DD sequences on the inner levels so that
recursive nesting of DD is possible.

\begin{proof}

The evolution under the control of $\Omega$ reads
\begin{equation}
U(T)=\Omega^{N}V_{N}\Omega V_{N-1}\cdots\Omega V_{1}\Omega V_{0},
\end{equation}
with the evolution operator
\begin{align}
V_{n} & \equiv\mathcal{T}\exp\left[-i\int_{T_{n}}^{T_{n+1}}H(t)dt\right]\nonumber \\
 & \equiv\mathcal{T}_{\theta}\exp\left[-i\tau_{n}\int_{\frac{n}{N+1}\pi}^{\frac{n+1}{N+1}\pi}H^{\text{rel}}(\theta)d\theta\right],
 \end{align}
where $\mathcal{T}_{\theta}$ stands for ordering in $\theta$, and
\begin{equation}
H^{\text{rel}}(\theta)=\frac{N+1}{\pi}H(t),
\end{equation}
where $\theta=\frac{n\pi}{N+1}+\frac{t-T_n}{\tau_n}\frac{\pi}{N+1}$
for $t\in(T_n,T_{n+1}]$. The symmetry requirements given in
Eq.~(\ref{eq:HtSCondit}) are transformed to
\begin{equation}
H^{\text{rel}}(\frac{n\pi}{N+1}+\theta)=H^{\text{rel}}(\frac{n+1}{N+1}\pi-\theta)=H^{\text{rel}}(\theta).
\label{eq:H_symetric}
\end{equation}
The Hamiltonian $H^{\text{rel}}(\theta)$ can be
separated into two parts,
\begin{equation}
H^{\text{rel}}(\theta)=C(\theta)+A(\theta),
\end{equation}
with
\begin{subequations}
\begin{align}
C(\theta)=\left[H^{\text{rel}}(\theta)+\Omega H^{\text{rel}}(\theta)\Omega\right]/2,\\
A(\theta)=\left[H^{\text{rel}}(\theta)-\Omega
H^{\text{rel}}(\theta)\Omega\right]/2.
\end{align}
\end{subequations}
$C(\theta)$ and $A(\theta)$ commute and anti-commute with the
operator $\Omega$, respectively.

Now we rewrite the propagator as
\begin{equation}
U(T)=\mathcal{T}_{\theta}\exp\left[-iT\int_{0}^{\pi}G(\theta)\Big(C(\theta)+F(\theta)A(\theta)\Big)d\theta\right],\label{eq:U_ThetaForm}
\end{equation}
where
\begin{subequations}
\begin{align}
G(\theta)&=\frac{1}{2}\left[\cos\frac{n\pi}{N+1}-\cos\frac{(n+1)\pi}{N+1}\right],\\
F(\theta)&=(-1)^{n},
\end{align}
\end{subequations}
for $\theta\in\left(\frac{n\pi}{N+1},\frac{(n+1)\pi}{N+1}\right]$.
Thus, the part of Hamiltonian $C(\theta)$ that commutes with
$\Omega$ is modulated by the step function $G(\theta)$ which has
step heights given by the UDD intervals, and the part of Hamiltonian
$A(\theta)$ that anti-commutes with $\Omega$ is modulated by
$G(\theta)$ and the periodic modulation function $F(\theta)$.
Furthermore, both $C(\theta)$ and $A(\theta)$ have the same
symmetries as $H^{\text{rel}}(\theta)$ in Eq.~(\ref{eq:H_symetric}).
The symmetries of the time-dependent Hamiltonian and the modulation
functions make them have particular Fourier expansions, which lead
us to a proof of the theorem in a procedure similar to the proof of
UDD in Ref.~\cite{Yang2008PRL}.

The Fourier expansions of the modulation functions and the time-dependent Hamiltonians are
\begin{subequations}
\begin{align}
G(\theta)&=\sum_{k=0}^{\infty}g_{k}\sin[2k(N+1)\theta\pm\theta],\label{eq:G_Theta} \\
F(\theta)&=\sum_{k=0}^{\infty}f_{k}\sin[(2k+1)(N+1)\theta],\label{eq:F_Theta}\\
C(\theta)&=\sum_{k=0}^{\infty}c_{k}\cos[2k(N+1)\theta],\label{eq:C_Theta}\\
A(\theta)&=\sum_{k=0}^{\infty}a_{k}\cos[2k(N+1)\theta].\label{eq:Z_Theta}\end{align}
\end{subequations}
Here the operators $c_{k}$ and $a_{k}$ commute and anti-commute with
$\Omega$, respectively. The features of these Fourier expansions to
be used in the proof below are: (i) Both $C$ and $A$ contain only
cosine harmonics of order of even multiple of $(N+1)$; (ii)
$F(\theta)$ contains only sine harmonics of order of odd multiple of
$(N+1)$; $G(\theta)$ contains only sine harmonics of an order
differing from an even multiple of $(N+1)$ by $+1$ or $-1$.

With the product-to-sum trigonometric formulae, we have
\begin{equation}
U(T)=\mathcal{T}_{\theta}\exp\left[-iT\int_{0}^{\pi}\left(\tilde{C}(\theta)+\tilde{A}(\theta)\right)d\theta\right],\end{equation}
 with\begin{subequations} \begin{align}
\tilde{C}(\theta)\equiv& G(\theta)C(\theta)=\sum_{k}\tilde{c}_{k}\sin[2k(N+1)\theta\pm\theta],\\
\tilde{A}(\theta)\equiv& G(\theta)F(\theta)A(\theta)\nonumber\\
=&\sum_{k}\tilde{a}_{k}\cos[(2k+1)(N+1)\theta\pm\theta].
\end{align}
\end{subequations}
A straightforward method is to expand $U(T)$ according to the
standard time-dependent perturbation theory.  It should be noted
that such perturbation-theoretic expansion requires that the
modulated Hamiltonian have bounded norm. In the expansion, the terms
which do not commute with $\Omega$ must contain an odd times of
$\{\tilde{a}_{k}\}$ (since $\Omega$ anti-commutes with
$\{\tilde{a}_{k}\}$). The expansion coefficients can be written as
\begin{align}
& (-iT)^{n}\times \nonumber \\
&\int_{0}^{\pi}y_{k_{1}}^{\alpha_1,\eta_{1}}(\theta_{1})\int_{0}^{\theta_{1}}y_{k_{2}}^{\alpha_2,\eta_2}(\theta_{2})\cdots
\int_{0}^{\theta_{n-1}}y_{k_{n}}^{\alpha_n,\eta_{n}}(\theta_{n})d\theta_{1}\cdots
d\theta_{n},
\label{eq:integralVanish}
\end{align}
with $y_{k}^{s,\pm}(\theta)\equiv\sin[2k(N+1)\theta\pm\theta]$
associated with an operator $\tilde{c}_{k}$, and
$y_{k}^{c,\pm}(\theta)\equiv\cos[(2k+1)(N+1)\theta\pm\theta]$
associated with an operator $\tilde{a}_k$, for $\alpha_j\in\{c,s\}$
and $\eta_{j}\in\{+,-\}$. By induction and repeatedly using the
product-to-sum trigonometric formulae, one can straightforwardly
verify that the coefficients in Eq.~(\ref{eq:integralVanish}) vanish
for $n\leq N$ and $y_{k}^{c,\pm}$ appearing an odd number of times.
Thus vanish any terms in the expansion which contain products of an
odd number of operators in $\{\tilde{a}_k\}$ and have a power of $T$
lower than $(N+1)$.

\end{proof}

\subsubsection{NUDD}

For a time-independent Hamiltonian $H$ under DD control of
instantaneous operations of a unitary Hermitian operator $\Omega$
applied at $t_1,t_2,\ldots,t_{N'}$, the evolution $U(\tau)$ from
$t_0=0$ to $t_{N'+1}=\tau$
 is equivalent to the evolution under a time-dependent Hamiltonian
\begin{equation}
H(t)=\Omega^nH\Omega^n,\label{eq:HtStep}
\end{equation}
for $t\in(t_n,t_{n+1}]$. Such time-dependence is not analytic. If
$N'$ is an even number and  the operation sequence is symmetric, the
time-dependent Hamiltonian $H(t)$ is time symmetric in $[0,\tau]$.
Thus, according to Theorem~\ref{thm:Ht}, in a UDD sequence of an
operator $\Omega$ applied at $T_1,T_2,\ldots,T_N$ between $T_0=0$
and $T_{N+1}=T$, each interval of free evolution
$e^{-iH(T_{n+1}-T_n)}$ can be substituted with the evolution
inserted by a sequence of another operation $\Omega'$ applied at
$T_{n,1},T_{n,2},\ldots,T_{n,N'}$ between $T_{n,0}\equiv T_n$ and
$T_{n,(N'+1)}\equiv T_{n+1}$, with the same symmetric structure in
all intervals, i.e.,
\begin{subequations}
\begin{align}
T_{n+1}-T_{n,N'-k}&=T_{n,k+1}-T_n, \\
\frac{T_{n,k}-T_{n,k'}}{T_{n+1}-T_n} &
=\frac{T_{m,k}-T_{m,k'}}{T_{m+1}-T_{m}}.
\end{align}
\end{subequations}
In particular, the inner level control of $\Omega'$ can be chosen as
an even order UDD.

Now we describe the construction of NUDD for protecting a set of
unitary Hermitian operators $\{\Omega_l\}_{l=1}^{L}$. First, the
$N_{L}$th order UDD sequence of $\Omega_L$ is constructed with
pulses applied at \begin{equation}T_{n_L}
=T\sin^2\frac{n_L\pi}{2N_L+2},\end{equation} between $T_0=0$ and
$T_{N_{L}+1}=T$ as the outermost level of control. $N_L$ could be
either odd or even. Then the free evolution in each interval is
substituted by the $N_{L-1}$th order UDD sequence of $\Omega_{L-1}$
applied at
\begin{equation}
T_{n_L,n_{L-1}}=T_{n_L}+\left(T_{n_L+1}-T_{n_L}\right)\sin^2\frac{n_{L-1}\pi}{2N_{L-1}+2},
\end{equation}
in each interval between $T_{n_L,0}\equiv T_{n_L}$ and
$T_{n_L,N_{L-1}+1}\equiv T_{n_L+1}$, with $N_{L-1}$ being an {\em
even} number. So on and so forth, the $l$th level of control is
constructed by applying $N_l$ times of $\Omega_l$ in each interval
between $T_{n_L,\ldots,n_{l+1},0}\equiv
T_{n_L,\ldots,n_{l+2},n_{l+1}}$ and
$T_{n_L,\ldots,n_{l+1},N_{l}+1}\equiv
T_{n_L,\ldots,n_{l+2},n_{l+1}+1}$ at
\begin{align}
T_{n_L,\ldots,n_{l+1},n_{l}}&
=T_{n_L,\ldots,n_{l+1}}
\nonumber \\ &
+\left[T_{n_L,\ldots,n_{l+2},n_{l+1}+1}-T_{n_L,\ldots,n_{l+1}}\right]\sin^2\frac{n_l\pi}{2N_l+2},
\end{align}
with $N_l$ being an {\em even} number. We denote the evolution under
such NUDD as
$U^{\text{U}[N_1,N_2,\ldots,N_L]}_{\Omega_1,\Omega_2,\ldots,\Omega_L}(T)$,
where the superscript $\text{U}$ denotes the nesting of UDD
sequences.

According to Theorem~\ref{thm:Ht}, the outer levels of UDD control
are not affected by the inner levels of even-order UDD control. And
according to Theorem~\ref{thm:MOOS}, the inner levels of DD control
are not affect by the outer levels of control since $\{\Omega_l\}$
is an MOOS. Thus each operator $\Omega_l$ is protected up to an
error of $O\left(T^{N_l+1}\right)$. The number of control intervals
is
\begin{equation}
N^{\text{U}[N_1,N_2,\ldots,N_L]}_{\text{pulse}}=\left(N_1+1\right)\left(N_2+1\right)\cdots\left(N_L+1\right),
\label{Eq:PulseNumber}
\end{equation}
increasing polynomially with the decoupling order.

\section{NUDD of multi-qubit systems}
\label{sec:Qubit}

\subsection{General multi-qubit systems}

To protect a multi-qubit system to a given order of precision, we
just need to protect the operators in an MOOS described in
Sec.~\ref{sub:MOOS}, by using the method depicted in
Sec.~\ref{sub:FirstOrder} for the first order preservation, or by
using the NUDD scheme in Sec.~\ref{sub:NestedDD} for higher order
preservation.
 Since for general multi-qubit systems, the MOOS can be chosen as a set of single-qubit
operators such as the Pauli matrices
$\{\sigma^{(l)}_x,\sigma^{(l)}_z\}_{l=1}^{L}$, NUDD can be
implemented with only single-qubit flips.

In suppressing the relaxation or pure dephasing to the first order,
i.e., in protecting the MOOS $\{\sigma_{z}^{(l)}\}_{l=1}^{L}$ or
$\{\sigma_{x}^{(l)}\}_{l=1}^{L}$, the first order scheme in
Sec.~\ref{sub:FirstOrder} requires $2^{L}$ pulse intervals. For
suppressing decoherence in general cases, the first order scheme
requires $4^{L}$ intervals to protect the MOOS
$\{\sigma_{x}^{(l)},\sigma_{z}^{(l)}\}_{l=1}^{L}$. Such numbers of
intervals are actually the minima required for protecting $L$-qubit
systems, as proven in Ref.~\cite{Wocjan02QIC}. In preserving the
coherence of a multi-qubit system to an arbitrarily high order of
precision by NUDD, the number of intervals given in
Eq.~(\ref{Eq:PulseNumber}) increases polynomially with the
decoupling orders, much less than that required in CDD. A question
is whether NUDD is optimal or nearly optimal in terms of the number
of control pulses. For QDD control of one qubit, numerical check up
to the fourth order indicates that NUDD is nearly optimal, differing
from the optimal solutions by less than 3 control
pulses~\cite{West2010}. NUDD of a larger MOOS, however, can be shown
to be far from the optimal in the second decoupling order: For an
$L$-qubit system suffering pure dephasing, an $L$-level NUDD in the
second order requires $\sim 3^L$ control intervals, while the SDD,
which uses two mutually symmetric first order DD sequences to
realize the second order control (see Sec.~\ref{sub:GeneralCDD}),
requires only $2\times 2^L$ intervals. We expect that in higher
orders of DD, there exist DD schemes (using only single-qubit
control for multi-qubit systems) much more efficient than NUDD. But
no explicit solutions are known to us, except for a few numerical
solutions.

\subsection{Discussions on QDD}

In particular, for a single-qubit system, NUDD reduces to QDD with
even order DD on the inner level. Explicitly the nested sequence
$U_{\sigma_{z},\sigma_{x}}^{\text{U}[2N_1,N_2]}(T)$ is the QDD
sequence protecting the MOOS operators $\sigma_{z}$ and $\sigma_{x}$
to orders $2N_1$ and $N_2$, respectively. Thus based on
Theorems~\ref{thm:MOOS} and ~\ref{thm:Ht}, the validity of QDD with
even order UDD on the inner level is proven.

But Theorem~\ref{thm:Ht} does not apply to the case of odd order UDD
control on the inner level. Actually, when the inner level UDD has
an odd order, which breaks the symmetry condition of the theorem,
the outer level UDD may be spoiled. For a specific example, let us
consider the control of a Hamiltonian like
$$H=J_{0}+J_{1}\Omega_{1}+J_{2}\Omega_{2}+J_{1,2}\Omega_{1}\Omega_{2},$$
where $J_{0}$, $J_{1}$, $J_{2}$, and $J_{1,2}$ are arbitrary bath
operators, and $\Omega_{1}$ and $\Omega_{2}$ are two system
operators forming an MOOS (such as $\Omega_{1}=\sigma_{z}$ and
$\Omega_{2}=\sigma_{x}$ for a single-qubit system). We choose the
inner level control as the first order UDD of $\Omega_{1}$ and the
outer level as the second order UDD of $\Omega_{2}$. The propagator
of this NUDD is
\begin{align}
U_{\Omega_{1},\Omega_{2}}^{\text{U}[1,2]}(T) = &
\left(\Omega_{1}e^{-iH\tau}\Omega_{1}e^{-iH\tau}\right)
\Omega_{2}\left(\Omega_{1}e^{-iH2\tau}\Omega_{1}e^{-iH2\tau}\right)\Omega_{2}
\nonumber \\
 &
\times\left(\Omega_{1}e^{-iH\tau}\Omega_{1}e^{-iH\tau}\right),
\end{align}
 where the evolution in each pair of parentheses corresponds to the UDD
 control on the inner level and $\tau=T/8$.
 The time expansion gives
\begin{align}
U_{\Omega_{1},\Omega_{2}}^{\text{U}[1,2]}(T) =
 U_{B}+\Omega_{1}O(T^{2})+\Omega_{1}\Omega_{2}O(T^{2}) +O(T^{3}),
 \end{align}
 where $U_{B}$ is a pure bath evolution operator.
 Thus even though a second order UDD
sequence of $\Omega_{2}$ is applied, $\Omega_{2}$ is preserved
only to the first order.

The above example indicates that the effective Hamiltonian resulting
from the inner level UDD control can not be written into an analytic
form. Otherwise, according to Ref.~\cite{Pasini2010UDDTime}, which
establishes the performance of UDD on analytically time-dependent
systems, the outer level UDD should not be affected. As proposed in
Ref.~\cite{Pasini2010UDDTime}, for an NUDD evolution
$$U_{\Omega_1,\Omega_2}^{U[N',N]}(T)=\Omega_2^{N}U_{\Omega_1}^{U[N']}(\tau_{N})\Omega_2
\cdots
U_{\Omega_1}^{U[N']}(\tau_{1})\Omega_2U_{\Omega_1}^{U[N']}(\tau_{0}),
$$
one can define the effective Hamiltonian resulting from the inner
level of control as
\begin{equation}
\tilde{H}_{\text{eff}}(T_n+t)\equiv i
\left[{\partial_t}U_{\Omega_1}^{U[N']}(t)\right]\left[U_{\Omega_1}^{U[N']}(t)\right]^{\dagger},
\label{Eq:effH}
\end{equation}
for $t\in(0,\tau_n]$. The outer level UDD can be viewed as acting on
this effective Hamiltonian. The Hamiltonian defined in
Eq.~(\ref{Eq:effH}), however, is only piecewise analytic and is even
discontinuous at $T_n$'s. Therefore, the theorem established in
Ref.~\cite{Pasini2010UDDTime} about UDD control of analytically
time-dependent systems does not apply to QDD.

Thus, the complete proof (or disproof) of the validation of QDD is
still an open question. We should mention that if QDD is proven
valid for one qubit, according to our formalism of nested DD, the
same NUDD is also valid for any two operators forming an MOOS. In
this way the QDD control can be generalized, in particular, to the
protection of two-qubit systems from dephasing and disentanglement,
etc.

\section{DD by pulses of finite amplitude\label{sec:SoftPulse}}

The ideal instantaneous pulses are not realistic in experiments
since they would contain an infinite amount of energy. Of course,
when the pulses are mush shorter than the other timescales of the
system and the bath, it is a good approximation to treat them as
infinitely short. But if this condition is not satisfied, it is of
interest to consider DD by pulses of finite amplitude. The problem
of first order DD with finite-amplitude pulses has been considered
within the Eulerian DD framework~\cite{Viola2003EulerianDD} and in a
geometric picture~\cite{Chen2006GeometricDD}. It is possible to
achieve arbitrary control precision by recursive construction of
pulse shapes~\cite{Khodjasteh2010}. DD of single-qubit systems using
finite-amplitude pulses up to a control error in the second order of
pulse durations has been presented in
Refs.~\cite{PasiniPRA2008,Pasini2008ShortPulse}. In UDD,
finite-amplitude pulses of higher orders of control precision can
also be incorporated~\cite{Uhrig2010}.

Here we consider the general case of DD by finite-amplitude pulses.
Let us consider a short-pulse operation by the Hamiltonian
\begin{equation}
H_{\Omega}(t)=v(t)\Omega,
\end{equation}
where $\Omega$ is a unitary Hermitian operator. We aim to design the
pulse shape of $v(t)$ such that the evolution during the pulse
control approximates the ideal $\delta$-pulse control up to a
certain order of the pulse duration $\tau_p$, i.e.,
\begin{align}
U(\tau_{p},0) & =\mathcal{T}\exp\left[-i\int_{0}^{\tau_{p}}\left[H+H_{\Omega}(t)\right]dt\right]
\nonumber\\
 & = e^{-i(\tau_{p}-\tau_{s})H}P_{\Omega}e^{-i\tau_{s}H}+O(\tau_{p}^{M_{p}}),
\end{align}
where $P_{\Omega}\equiv
\exp\left[-i\int_{0}^{\tau_{p}}H_{\Omega}(t)dt\right]$ is the
desired instantaneous control applied at the time $\tau_{s}$. In
particular, we need the $\pi$ pulse $P_{\Omega}=\Omega$ up to a
trivial global phase factor. Unfortunately, a no-go theorem
established in Ref.~\cite{PasiniPRA2008,Pasini2008ShortPulse}
restricts that instantaneous $\pi$ pulses can not be approximated by
a finite-amplitude pulse with error lower than
$O\left(\tau_{p}^2\right)$ without perturbing the bath evolution.
Thus, here we focus on the first order pulse shaping with $M_{p}=2$.

We write the evolution operator as
\begin{equation}
U(\tau_{p},0)=e^{-i(\tau_{p}-\tau_{s})H}U_{\Omega}e^{-i\tau_{s}H},\end{equation}
where\begin{equation}
U_{\Omega}=\mathcal{T}\exp\left[-i\int_{0}^{\tau_{p}}\tilde{H}_{\Omega}(t)dt\right],\end{equation}
with $\tilde{H}_{\Omega}(t)\equiv
e^{iH(t-\tau_s)}H_{\Omega}(t)e^{-iH(t-\tau_s)}$. The correction term
is
\begin{align}
h_{\Omega}(t) & \equiv \tilde{H}_{\Omega}(t)-H_{\Omega}(t)
\nonumber \\
 & =
v(t)\sum_{k=1}^{\infty}\frac{(t-\tau_s)^k}{k!}\underbrace{[iH,[iH,\cdots[iH,\Omega]\cdots]]}_{k
\textrm{~folds}}.
\end{align}
We want to design $v(t)$ such that the control error
\begin{align}
\delta P_{\Omega}& \equiv
U_{\Omega}-P_{\Omega}
\nonumber \\ &
=\mathcal{T}\left\{e^{-i\int_{0}^{\tau_{p}}H_{\Omega}(t)dt}\left[e^{-i\int_{0}^{\tau_{p}}h_{\Omega}(t)dt}-1\right]\right\}  =O(\tau_{p}^{2}).
\end{align}
The leading order term in $\delta P_{\Omega}$ is
\begin{align}
\eta^{(1)} & =\mathcal{T}\left\{e^{-i\int_{0}^{\tau_{p}}H_{\Omega}(t)dt}\int_{0}^{\tau_{p}}(t-\tau_s)v(t)[H,\Omega]dt\right\}
\nonumber \\
 & =\int_{0}^{\tau_{p}}(t-\tau_s)v(t)e^{-i\int_{t}^{\tau_{p}}H_{\Omega}(s)ds}[H,\Omega]e^{-i\int_{0}^{t}H_{\Omega}(s)ds}dt.
\label{eq:eta1}
\end{align}
Using $\Omega^{2}=1$ and $H_{\Omega}(t)=v(t)\Omega$, we have
\begin{equation}
e^{-i\int_{t_{1}}^{t_{2}}H_{\Omega}(s)ds}=\cos\left[\int_{t_{1}}^{t_{2}}v(s)ds\right]
-i\Omega\sin\left[\int_{t_{1}}^{t_{2}}v(s)ds\right].
\label{eq:OmegaRoation}
\end{equation}
Now we decompose the Hamiltonian $H$ into two parts as $H=A+C$, with $A$ and
$C$ anti-commuting and commuting with $\Omega$, respectively.
The leading order error term in Eq.~(\ref{eq:eta1}) becomes
\begin{equation}
\eta^{(1)}=[A,\Omega]\eta_{11}-i\Omega[A,\Omega]\eta_{12},\end{equation}
where
\begin{subequations}
\begin{align}
\eta_{11}=\int_{0}^{\tau_{p}}(t-\tau_s)
v(t)\cos\left[\phi_{0}-\psi(t)\right]dt,\\
\eta_{12}=\int_{0}^{\tau_{p}}(t-\tau_s)v(t)\sin\left[\phi_{0}-\psi(t)\right]dt,
\end{align}
\end{subequations}
 with
$\psi(t)\equiv2\int_{\tau_{s}}^{t}v(s)ds$ and
$\phi_{0}=\int_{\tau_{s}}^{\tau_{p}}v(s)ds-\int_{0}^{\tau_{s}}v(s)ds$.
To eliminate the leading order error, we just need to make
$\eta_{11}=\eta_{12}=0$, which are the same as those derived in
Ref.~\cite{PasiniPRA2008} for single-qubit flip control.

Using the finite-amplitude pulses designed as depicted above, we can
realize DD based on an MOOS up to an error in the second order of
the pulse duration. Note that in some DD schemes, such as CDD
$U_{\Omega_1,\ldots,\Omega_L}^{C[N_1,\ldots,N_L]}(T)$ as shown in
Eq.~(\ref{eq:CDDway2}), operations of different $\Omega_l$'s may
coincide. For example, $\Omega_1$, $\Omega_2$, and $\Omega_3$
coincide at the end of the sequence
$U_{\Omega_1,\Omega_2,\Omega_3}^{C[N_1,N_2,N_3]}(T)$. In this case,
we can define a new unitary Hermitian operator $\Omega'$ as the
product of the operations (such as $\Omega_1\Omega_2\Omega_3$ in the
example), and design the pulse $H_{\Omega'}(t)=v(t)\Omega'$ to
achieve the operation up to an error of $O(\tau_p^2)$. Such
finite-amplitude pulse operation, however, would involve multi-qubit
interactions and may not be easy to be implemented in experiments,
unless $\Omega'$ happens to be a single-qubit operation. This
problem, fortunately, does not exist in NUDD, since there no two
operations coincide, which stands for another advantage of NUDD over
CDD.

Ref.~\cite{Uhrig2010} has presented a method to implement UDD by
higher order shaped pulses. At the first sight, it seems that those
pulses can be incorporated in NUDD. However, the method in
Ref.~\cite{Uhrig2010} requires a starting and a stopping pulse to
protect the UDD sequence. In NUDD, such starting and stopping
operations will be mixed up with the outer level control applied at
the same time. Then the operations on the outer levels need to be
redesigned, which may be much more complicated than the design in
Ref.~\cite{Uhrig2010} since different operations may interfere with
each other and multi-qubit interactions may be involved. Explicit
implementation of NUDD with finite-amplitude pulses of higher order
control accuracy is an interesting topic for future work.

\section{Conclusions and discussions\label{sec:Conclusions}}

Based on two theorems, we have presented explicit schemes of
dynamical decoupling to preserve operators in an MOOS (i.e., unitary
Hermitian operators which either commute or anti-commute with each
other) to arbitrary decoupling orders for quantum systems
arbitrarily coupled to quantum baths. All system operators in a Lie
algebra generated from the MOOS by commutation, anti-commutation,
linear combinations, and repetitions are also preserved.
Theorem~\ref{thm:MOOS} states that the inner levels of DD control
are unaffected by the outer levels if the control operations are
elements of an MOOS. Theorem~\ref{thm:Ht} states that UDD still
works if the Hamiltonians in different intervals have the same
function form of the relative time and are symmetric, regardless of
the analytic properties of the Hamiltonians. These theorems enable a
construction of higher order DD by nesting UDD sequences of even
orders. NUDD protects system operators in a Lie algebra generated
from an MOOS to an arbitrary order of precision. For multi-qubit
systems, any physical quantities can be protected, and NUDD can be
implemented by single-qubit operations. For single-qubit systems,
NUDD reduces to QDD with even order UDD on the inner level. Thus the
theorems provide a rigorous proof of the validity of QDD with even
order DD on the inner level.

NUDD achieves a desired decoupling order with only a polynomial
increase in the number of pulses, with exponential saving of the
number of pulses as compared with CDD of the same decoupling order.
In suppressing the general decoherence, the number of pulses still
scales exponentially with the number of qubits. Such exponential
increase, indeed, is required by a theorem which sets the minimum
number of control intervals to be $4^L$ or $2^L$ for protecting a
general or pure dephasing $L$-qubit system to the first decoupling
order, respectively~\cite{Wocjan02QIC}.

For Hamiltonians of certain structures, such as the Hamiltonians of
qudit systems with bipartite interactions, reduction in the number
of pulses is possible~\cite{Wocjan02PRA}. Fewer levels of nesting
are required if the structures of the Hamiltonians are exploited and
a proper MOOS is designed. The number of pulses can also be greatly
reduced if we protect only some logically encoded qubits or some
particular states of the
system~\cite{Viola2002EncodedDD,Byrd2002PRL,Dhar2006_PRL,MultilevelUDD:2010}.
For example, if we choose the MOOS as
$\{\sigma^{(1)}_x\otimes\sigma^{(2)}_x, \sigma^{(1)}_z,
\sigma^{(2)}_z\}$ for a two-qubit system, the only possible noise
generator~\cite{Viola2003EulerianDD} after protection is
$\sigma^{(1)}_z\otimes\sigma^{(2)}_z$, which commutes with all the
elements in the MOOS; in this case, the logical qubit $\alpha
|\uparrow\uparrow\rangle+\beta |\downarrow\downarrow\rangle$ or
$\alpha' |\uparrow\downarrow\rangle+\beta'
|\downarrow\uparrow\rangle$ is protected by only three levels of
nesting.

NUDD protecting two operators forming an MOOS is near-optimal and
has the same timing as QDD. In general, however, NUDD is by far not
optimal, since a large nesting level $L$ requires much more control
intervals than the symmetrized DD for achieving the second
decoupling order. An interesting question for future study is how to
construct optimal or nearly optimal higher-order DD for general
multi-qubit or multi-level systems.

For realistic implementation of DD, we have derived the conditions
for finite-amplitude pulses to simulate ideal operations up to an
error in the second order of pulse duration, and the conditions
reach the same results as for single-qubit flip control given in
Ref.~\cite{PasiniPRA2008}. Thus we can apply the pulses designed in
Ref.~\cite{PasiniPRA2008} to the higher order DD schemes for general
quantum systems.

\begin{acknowledgments}
This work was supported by Hong Kong GRF CUHK402209. We are grateful
to L.Viola and D. Lidar for discussions.
\end{acknowledgments}


\end{document}